\documentclass[11pt]{article} 

\usepackage{amsbsy,amsmath,amsthm,amsfonts, amssymb} 
\usepackage{sectsty}

\makeatletter

\newtheorem{theorem}{Theorem}

\usepackage{pdfpages}

\usepackage{natbib}

\newcommand{\logit}{\mbox{logit}}

\newcommand{\tausq}{\tau^2}

\newcommand{\sis}{\sigma^2}

\newcommand{\bX}{ {\bf X} }

\newcommand{\bY}{ {\bf Y} }

\newcommand{\bpsi}{ \mbox{\boldmath $ \psi $} }

\newcommand{\pr}{\mbox{pr}}

\newcommand{\beq}{\begin{equation}}
\newcommand{\eeq}{\end{equation}}
\newcommand{\bea}{\begin{eqnarray}}
\newcommand{\eea}{\end{eqnarray}}
\newcommand{\beas}{\begin{eqnarray*}}
\newcommand{\eeas}{\end{eqnarray*}}

\newcommand{\nn}{\nonumber}
\newcommand{\bex}{\begin{example} \rm }
\newcommand{\eex}{\rule{5pt}{5pt} \end{example}}
\newcommand{\bdf}{\begin{definition} \rm }
\newcommand{\edf}{\end{definition}}
\newcommand{\bthm}{\begin{theorem} \rm \sf }
\newcommand{\ethm}{\end{theorem}}
\newcommand{\bproof}{\begin{proof} \rm \sf }
\newcommand{\eproof}{\end{proof}}
\newcommand{\balg}{\begin{algorithm} \rm }
\newcommand{\ealg}{\end{algorithm}}

\newcommand{\bth}{\begin{theorem}}

\begin{document}

\thispagestyle{empty}


\raggedright

\newpage

\today


{\bf PERCENTILE-BASED RESIDUALS FOR MODEL ASSESSMENT}\\

Sophie B\'erub\'e\footnote{Contact author: sberube3@jhmi.edu}, Abhirup Datta, Qingfeng Li, Chenguang Wang, Thomas A. Louis\\
Johns Hopkins Bloomberg School of Public Health, 615 N. Wolfe St., Baltimore MD 21205

\section*{Abstract}
Residuals are a key component of diagnosing model fit. The usual  practice is to compute standardized residuals  using expected values and standard deviations of the observed data, then  use these values to detect outliers and assess model fit.  Approximate normality of these  residuals is    key  for this process to have good properties, but in many modeling contexts, especially for complex, multi-level models, normality may not hold. In these cases outlier detection and   model diagnostics  aren't  properly calibrated. Alternatively, as  we demonstrate, residuals computed from the percentile location of a datum's value in its full predictive distribution lead to well calibrated evaluations of model fit. We generalize an approach described by~\citet{DunnandSmyth:etal:1996} and evaluate properties mathematically, via case-studies and by simulation. In addition, we show that the standard residuals can be calibrated to mimic the percentile approach, but that this extra step is avoided by directly using percentile-based residuals.  For both the percentile-based residuals and the  calibrated standard residuals, the use of full predictive distributions with the appropriate location, spread and shape is necessary for valid assessments.

{\bf KEYWORDS}\\ Percentile-based residuals, Model assessment, Outlier detection, non-Gaussian predictions, Well-calibrated diagnostics.

\baselineskip 16.5pt

\section{Introduction}
Residuals are a  key component of diagnosing  model fit.  They are used to identify outlying data points, and plotted against predicted values they can reveal model  lack of fit.  The commonly used standard residuals are computed as, {\tt Residual = (Observed - Expected)/SD}, where the `Expected' and `SD'(standard deviation) come from a data point-specific  predictive distribution.    Irrespective of model form, if the predictive distributions are well estimated, these residuals have mean 0 and variance 1, however there is no guarantee that they will have a $N(0,1)$ distribution. If they aren't  $N(0,1)$, outlier detection and global model assessments can perform poorly, with poorly calibrated Type~I error for outlier detection, potentially low power, and misleading residual plots. 
Therefore, a  residual that is well-calibrated for any predictive distribution  has the potential to improve performance, and  a percentile-based approach achieves this goal.

Herein, we develop and evaluate a generalization of~\citet{DunnandSmyth:etal:1996}. Expanding on their randomized percentile-based residual, we consider an approach that goes beyond the first two moments of the full predictive distribution and uses all data points to derive the full predictive distribution, which is then used in its entirety to compute the percentile-based residuals. We further derive mathematical properties of these percentile-based residuals including their power.

Percentile-based residuals are computed by  finding the percentile location of an observation in its full predictive distribution,  then computing the corresponding Gaussian quantile  to produce the residual.    For continuous distributions, when the full predictive distribution matches the underlying truth, these residuals are distributed $N(0, 1)$. The definition is general in that the  full predictive distribution can be from a Bayesian analysis (including using the MCMC draws as the distribution), from  frequentist modeling, from machine learning (e.g. classification and regression trees (CART), support vector machines, etc.) or from any other modeling approach.  More recent work on this topic includes~\citet{cook:etal:2006} describing the `percentile, and  Gaussian quantile' approach to evaluating computer software and~\citet{efron:statsci:2008} who presents an example of transforming to z-values. The use of inverse Gaussian transformations was introduced much earlier than 2006 though, with~\citet{Efron:1987} showing an example of the ability of normalizing inverse transformations to automatically improve performance of the bootstrap confidence intervals without the user having to re-calibrate the process for each new application.  

With `SD', the standard deviation of the predictive distribution, if the working predictive distribution is Gaussian, the standard   (Observed - Expected)/SD residuals are identical to the percentile-based residuals. However, the Gaussian assumption is commonly inappropriate. For instance,~\citet{DunnandSmyth:etal:1996} consider a log-linear model in the context of survival data as well as a logistic regression in the context of bionomial data. Rather than these examples, we  consider models for which the full predictive distribution may incorporate  parameter uncertainty, and hierarchical models  with distributions that aren't Gaussian. We show that in such cases, assuming normality, and using the standard residuals can lead to poorly calibrated  model diagnostics and outlier detection; in the context of formal hypothesis testing,  conservative   or inflated Type~I errors and similar influences on the  power of the test.  Against this background, we  show that replacing the standard residuals by the percentile-based properly calibrates model assessment and testing. 
Similar advantages are associated with using percentile-based residuals in diagnostic plots.
Finally  we show that   the usual residuals can be calibrated to mimic the percentile-based approach, but this step is avoided by direct use of percentiles.

\section{Notation and Methods} \label{methods.sec}

Let $(Y_k, \bX_k)$ represent all  data (dependent variable, covariates) for the $k^{th}$ ($k = 1, \ldots, K$) sampling unit,  and 
let $(\bY, \bX)$ represent all data.  We focus on a scalar $Y_k$, which can be a unit-specific summary statistic.   The analyst produces a working model,
$[Y_k \mid \bX_k, \bpsi]_{\mbox{\tt wkng}}$
with covariates $\bX_k$ and parameters $\bpsi$ (all parameters; slopes, variances, variance components, etc.). Embedded in the working model are all modeling assumptions and data analytic choices. Examples include linear and logistic regression,  CART, random forests and other machine learning approaches (for these  $\bpsi$ represents the underlying algorithm's end result).
Data analysis produces the working predictive cumulative distribution function for unit $k$,
\bea \label{workingPredictive.eqn}
D_k(Y_k) &=& D_k(Y_k \mid \bX_k, \mbox{Analysis}),
\eea
but the  true  predictive cumulative distribution function for the $k^{th}$ unit is,
\bea \label{truePredictive.eqn}
 F_k(Y_k) &=&  F(Y_k \mid \bX_k).
 \eea
 $D_k$  is determined by  the the modeling approach, and producing it is  quite general.  The full posterior distribution of $\bpsi$  can be used to generate an in or out of sample predictive distribution for $Y_k$ (e.g., the collection of MCMC samples), a plug-in approach substituting $\hat \bpsi$ for $\bpsi$ with no attention to uncertainty in the estimate, or the end-result of a  machine-learning algorithm, with or without infusion of uncertainty.  

\subsection{The (O - E)/SD, or standard residuals} \label{Rstar.sec}
With $Y_k=y_k$, the standard (Observed - Expected)/SD   residuals are,
\bea \label{Rstar.eqn}
R^*_k &=& \frac{y_k-  \mu_k}{ \sigma_k} \\
 \mu_k &=& E_{D_k}(Y_k \mid \bX_k) \nn \\
 \sigma_k^2 &=& V_{D_k}(Y_k \mid \bX_k)  \nn 
\eea
An example of the above residual is linear regression, where $R^*_k = \frac{y_k - \hat{y_k}}{\sqrt{\mbox{MSE}}}$. In order to perform model diagnostics like outlier identification and goodness of fit, the empirical distribution of the $R_k^*$ is evaluated relative to the $N(0,1)$ distribution; plotting  $R_k^*$ versus $ \mu_k$ can identify the need for model enhancement. 
If $( \mu_k,  \sigma_k)$ are the true mean and standard deviation under $F_k$, then the $R_k^*$ have mean $0$ and variance $1$, but  the full distribution can  be far from Gaussian unless the $Y$s are Gaussian or approximately Gaussian due to the Central Limit Theorem (CLT).

\subsection{Percentile-based residuals} \label{Rddag.sec} 
The  standard residuals can be represented as,
	\begin{equation}\label{eq:std2}
	R^*_k = \Phi^{-1}\left \{ \Phi_{\mu_k,\sigma_k}(y_k) \right \},
	\end{equation} 
	where $\Phi_{\mu,\sigma}$ denotes the CDF, or cumulative density functon, of a $N(\mu,\sigma^2)$ distribution and $\Phi=\Phi_{0,1}$. 
	This framework supports generalizing the definition of a residual by using the full predictive distribution of $Y_k$. To relax dependence  on the CLT, we find  the percentile location of $Y_k =y_k$ in the working predictive distribution $D_k$ and map it to the  associated quantile of a $N(0,1)$ distribution.  Specifically, we define
\bea \label{Rddag.eqn}
R_k^\ddagger &=& \Phi^{-1} \left \{D_k(y_k) \right \} ~\mbox{ (for continuous $D_k$)}, \\
&=& \Phi^{-1} \left \{D_k(y_k) - 0.5  \pr_{D_k}(Y_k = y_k) \right \} ~\mbox{ (for discrete $D_k$)}.\nonumber
\eea

The `one-half correction' is needed to balance the computation for a discrete distribution.  For example, if $D_k$ puts all mass at a single point and $y_k$ is that point, the uncorrected $R_k^\ddag =  \infty$; the corrected (and correct) $R_k^\ddag =0$.  If the direct estimate, $y_k$, is equal to the largest value of the predictive distribution, the  correction brings $R^\ddag$  from infinity to a finite value.  Even for a continuous $D_k$, either $R^*$ or $R^\ddag$ can be $\pm \infty$,  for example if the observed value is beyond the support of the predictive distribution.  In these cases,  truncating the residual, for example at $\pm5.0$,  is often appropriate.

Equation~\ref{Rddag.eqn} is equivalent to   replacing $\Phi_{\mu_k,\sigma_k}$ in equation~\ref{eq:std2} with the  predictive distribution $D_k$. It is also evident that when $D_k$ is Gaussian, the standard residuals are identical to the percentile-based. Importantly, this approach allows the user to estimate the working predictive distribution, and thus the residuals, using all data points.  


\section{Properties} \label{properties.sec}
The  data analyst generates the working predictive distribution $D_k$ 
using data from $K$ units with $n_k$ observations in unit $k$. The $n_k$ can be small and so the unit-specific, direct estimates, $y_k$ may be far from Gaussian. Thus, assuming that the $R_k^*$  are $N(0,1)$ may induce  false positive and false negative rates far from the nominal values in assessing the relation between the observed and expected values.  More generally, when each $D_k = F_k$, the empirical distribution of $(R^*_1, \ldots, R^*_K)$ can deviate substantially from $N(0,1)$.

Dropping the subscript $k$ with the understanding that the distributions are $k-$specific, we evaluate properties under 
$H_0: F = D$ and under $H_1: F \ne D$.
A discrepancy between $F$ and $D$ detected under  $H_0$ is a Type~I error (a false positive) and failing to detect a discrepancy between $F$ and $D$ under $H_1$ is a Type~II error (a false negative). 
Discrepancies between $F$ and $D$ can be induced by various conditions including the use of incorrect parametric families, an incorrect mean model, the lack of uncertainty infusion or some combination of these. In analyzing residuals, the goal is to detect discrepancies between $F$ and $D$ while controlling  Type~I error, optimizing power and producing valid and informative residual plots. We show herein that depending on the choice of $D$ the use of standard residuals ($R^*$) can cause Type~I error rate to be inflated or conservative and the shape of the true predictive distribution, $F$, can be mis-represented. 
 
\subsection{Properties of $R^*$}
Let $\mu_0$ and $\sigma_0$ denote the mean and variance of $D$. 
We first show that the Type~I error for the standard residuals is not well calibrated.

\begin{theorem}\label{th:calib}[Type~I error] The Type~I error at level $\alpha$ for standard residuals, $R^*$ is given by:

\begin{equation}
\alpha^*(\alpha) =  \left\{ 
	\begin{array}{c}
		1- \Phi_{\mu_0,\sigma_0} \{ D^{-1}(1 -\alpha) \}\mbox{ for right sided test}\\
		\Phi_{\mu_0,\sigma_0} \{ D^{-1}(\alpha) \}\mbox{ for left sided test}  \label{calibrate}\\
			
		\mbox{unique root of: } 1 - D(\Phi^{-1}_{\mu_0,\sigma_0} (1 - x/2)) + D(\Phi^{-1}_{\mu_0,\sigma_0} (x/2)) - \alpha =0 \\
		\hspace*{-.2in}{\mbox{ for two sided test}}
	\end{array}
	\right.
\end{equation}
\end{theorem}

{\bf Proof:} We give the proof for the right sided test. Let $U^* = \Phi \left (\frac{Y - \mu_0}{\sigma_0} \right )$, 
then


\bea
\pr(U^* \ge 1 - t) &=& 1 - \pr \left \{\Phi \left (\frac{Y - \mu_0}{\sigma_0} \right ) \le 1 - t \right \} \nn \\ 
&=& 1 - \pr \left \{Y \leq \Phi^{-1}_{\mu_0, \sigma_0} (1-t) \right \} \nn \\
&=& 1 - D \left \{\Phi_{\mu_0 , \sigma_0}^{-1}(1-t) \right \}. \nn \eea

So, the effective Type~I error for the standard residuals is $\alpha^*(\alpha)  = 1 - D \left \{ \Phi_{\mu_0 , \sigma_0} ^{-1}(1 - \alpha)\right \}$.  \hfill \qed

The proof for the left-sided test is essentially the same as for the right-sided; the proof for the two-sided test combines the right and left  tail probabilities. 

From Theorem~\ref{th:calib}, it follows  that the Type~I error for a right-sided test using $R^*$ induces  the following relationships to the nominal level, $\alpha$:  
\bea
\mbox{Inflated} \iff D\{\mu_0+\sigma_0\Phi^{-1}(1-\alpha)\} < (1-\alpha)  \mbox{ or }  \Phi^{-1}_{\mu_0,\sigma_0}(1 - \alpha) <D^{-1}(1-\alpha)\nn \\
\mbox{Exact} \iff D\{\mu_0+\sigma_0\Phi^{-1}(1-\alpha)\} = (1-\alpha) \mbox{ or }  \Phi^{-1}_{\mu_0,\sigma_0}(1 - \alpha)  = D^{-1}(1-\alpha)\nn \\
\mbox{Conservative} \iff D\{\mu_0+\sigma_0\Phi^{-1}(1-\alpha)\} > (1-\alpha) \mbox{ or } \Phi^{-1}_{\mu_0,\sigma_0}(1-\alpha) > D^{-1}(1-\alpha) \nn \\
\eea
It is now easy to identify the  conditions that lead to inflated or conservative Type~I error. If $D$ places more probability mass on the right-tail than the Gaussian distribution, the standard residuals will produce inflated Type~I error. Conversely, if $D$ places less probability mass on the right-tail than the Gaussian distribution, then the standard residuals will produce conservative Type~I error. Similar conditions can be derived for left- and two-sided tests. 

The proof of Theorem \ref{th:calib} immediately identifies the raw
power of $R^*$ under $H_1$, where `raw' indicates that the power is not adjusted for the poorly calibrated Type~1 error.
\begin{theorem}[\textbf{Power for $R^*$}]
$R^*$ has power for rejection to the right with nominal right-sided Type~I error  $\alpha$,
\bea \label{powerR*.eqn}
POW^*_F(\alpha) &=& 1- F\{\mu_0 +  \sigma_0 \Phi^{-1}(1-\alpha)\} = 1- F(\Phi_{\mu_0,\sigma_0}^{-1}(1-\alpha))
\eea
\end{theorem}

\subsection{Properties of $R^\ddag$}
For   continuous $F$ and $D$,
\bea \label{distns.eqn}
 \pr_F (D(Y) \le u) &=& \pr \{Y \le D^{-1} (u)\} = F\{D^{-1} (u)\}\\
 &=& u 
\mbox{, if } D = F. \nn
\eea
Consequently, under $H_0$ when $D=F$,   $ R^\ddag \sim N(0,1)$, and 
Type~I error  using the right-sided rejection region $(\Phi^{-1}(1-\alpha),\infty)$ is  perfectly calibrated. 

Theorem~\ref{distnofRddag.thm} gives the full distribution of $R^\ddag$ for the continuous case.
\bthm[\textbf{Distribution of $R^\ddag$}] \label{distnofRddag.thm}
 For continuous $F$ and $D$ with densities $f$ and $d$, let $G^\ddag_F$ be the distribution of $R^\ddag$ with density $g^\ddag_F$ computed under $F$.  Then, if  $F$ is absolutely continuous wrt $D$,
\bea \label{distnofRddag.eqn}
G^\ddag_F (r) &=& F \left [ D^{-1} \left \{\Phi(r) \right \} \right ] \\
g^\ddag_F(r) &=& 
\phi(r) \frac{f \left [ D^{-1} \{\Phi(r) \}\right ]}{d\left [ D^{-1} \{\Phi(r)\} \right ]} . \nn
\eea
Consequently, if $D=F, g^\ddag_F(r) = \phi(r)$.
 \ethm

\bproof
\bea
R^\ddag &=& \Phi^{-1} \left \{ D(Y)\right \} \nn \\
\noalign{so,}
G^\ddag_F (r) &=& \pr (R^\ddag \le r)  = pr \left\{ \Phi^{-1}(D(Y)) \le r\right\} = pr \left \{ D(Y) \le \Phi(r)  \right \} \nn \\
&=&  \pr \left [ Y \le D^{-1} \left \{\Phi(r) \right \} \right ] 
=F \left [ D^{-1} \left \{\Phi(r) \right \} \right ] \nn
\eea
Taking the derivative wrt $r$ gives the density in equation~(\ref{distnofRddag.eqn}). 
\eproof 

From equation~(\ref{distnofRddag.eqn}), for $g^\ddag_F$ to be a Gaussian density, the ratio must  
be 1.0 (as it is under  $H_0$).  More generally, $g^\ddag_F(r)$ contains $\phi(r)$ as a multiplicative factor which can produce a Gaussian-like shape.

\subsection{Power comparisons} \label{powercomp.sec} 
We  begin by deriving  power (e.g., probability of  detecting an outlier)  for the right-sided test at level $\alpha$.

\begin{theorem}[\textbf{Power for $R^\ddag$}]
The right-side rejection power of $R^\ddag$ for a one-sided test of nominal (and actual) size $\alpha$ is,
\bea \label{powerRddag.eqn}
POW^\ddag_F(\alpha) &=& 1-F\{D^{-1} (1-\alpha)\}
\eea
\end{theorem}

\begin{proof}
Substitute $r=\Phi^{-1}(1-\alpha)$ in equation~\ref{distnofRddag.eqn}.
\end{proof}

Using equations~\ref{powerR*.eqn} 
and~\ref{powerRddag.eqn} various comparisons can be made between the standard residuals ($R^*$) and the percentile-based residuals ($R^\ddag$).  For a right-sided test, since $F$ is monotonically increasing, $R^*$ will have higher, equal or less power than $R^\ddagger$ depending on whether  $\Phi^{-1}_{\mu_0,\sigma_0}(1-\alpha)$ is less than, equal to, or greater than  $D^{-1}(1-\alpha)$. Combining the power comparisons with the results for the  Type~I error of $R^*$ in~(\ref{th:calib}) we  have the following result:
\begin{theorem}\label{th:sizepower}
For a right-sided test, the standard residuals have inflated, correct, or conservative Type~I error, and higher, equal, or lesser power than $R^\ddagger$ depending on whether $\Phi^{-1}_{\mu_0,\sigma_0}(1-\alpha)$ is less than, equal to, or greater than $D^{-1}(1-\alpha)$.
\end{theorem}
The result  shows why it is inappropriate to use the standard residuals. When $\Phi^{-1}_{\mu_0,\sigma_0}(1-\alpha) < D^{-1}(1-\alpha)$, $R^*$ may have higher power than $R^\ddagger$, but that apparent win is induced at least in part by  inflated Type~I error. 
On the other hand, when $\Phi^{-1}_{\mu_0,\sigma_0}(1-\alpha) > D^{-1}(1-\alpha)$, $R^*$ will have conservative Type~I error and lower power than $R^\ddag$.
By contrast, the $R^\ddag$ have properly calibrated Tyep~I error and valid power. Importantly, even with well calibrated percentile-based residuals, multiple testing corrections should be performed when appropriate.

\subsubsection{Calibrating $R^*$} 
The $R^*$ can be adjusted to have properly calibrated Type~I error and thus valid power.
 From Theorem~\ref{th:calib}, using a right-sided rejection region of the form $(\Phi^{-1}(1-\alpha^*),\infty)$, $R^*$ gives a Type~I error of $1 - D\{\Phi_{\mu_0 , \sigma_0}^{-1}(1-\alpha^*)\}$. Equating this to the nominal level $\alpha$ we have:
\begin{equation}\label{eq:calibrr}
\alpha^*=1-\Phi_{\mu_0,\sigma_0}\{D^{-1}(1-\alpha)\}.
\end{equation}
Constructing the right-sided rejection region based on the the $(1-\alpha^*)^{th}$ quantile now gives a test based on usual residuals with  calibrated Type~I.
As Theorem~\ref{th:calibpower} shows,  the power for the  calibrated $R^*$ equals the power for the percentile-based residuals. 
\begin{theorem}\label{th:calibpower} For the right-sided test, the power of the calibrated $R^*$ equals that of  $R^\ddag$.
\end{theorem}
\begin{proof} From equation~\ref{powerR*.eqn}, the power for the right-sided test based on $R^*$ using calibrated rejection region is given by $1- F\{\Phi_{\mu_0,\sigma_0}^{-1}(1-\alpha^*)\}$. Using the definition of $\alpha^*$ from equation~\ref{eq:calibrr}, this gives $1- F\{D^{-1}(1-\alpha)\}$, which is the power for the percentile-based residuals from equation~\ref{powerRddag.eqn}. 
\end{proof}
This power equivalence also holds for a left-sided test, but does not hold exactly for a two-sided.  

Theorems~\ref{th:sizepower} and~\ref{th:calibpower} illustrate the pitfalls of using the standard residuals and highlight the importance of using the full,  predictive distribution for model assessment, via either $R^\ddag$ or the calibrated $R^*$.

\section{Simulation study}
\label{sec:org8bbdfc2}
We conducted a simulation study to evaluate and compare properties of $R^*$ (both uncalibrated and calibrated)  and $R^\ddag$.

\subsection{Data generation}
\label{sec:org25a4a2c}
For 
\(Y_k \sim\) Beta$(a_k,b)$,  
 \(b \equiv 3\),  \( \log (a_k) = \beta_0 + \beta_1
    X_{k,1} + \beta_2 X_{k,2},\) \(X_{k,1} \sim N(0,1)\) and
\(X_{k,2} \sim \mbox{Bernoulli}(0.5)\), we consider two \emph{true} models, \(F_k^{(0)}\) and \(F_k^{(1)}\). For each model, \(\beta_0=0, \beta_1=1\). 
For \(F_k^{(0)}\),  \(\beta_2 = 0\); for \(F_k^{(1)}\), \(\beta_2 = -5\).
\(F_k^{(0)}\) and \(F_k^{(1)}\)
can be interpreted as the  cumulative distribution under the null and alternative hypotheses, respectively. 

In the  working model ($D_k$), $(b, \beta_0, \beta_1)$ are unknown,  
and $\log(a_k) = \beta_0 + \beta_1X_{k,1}.$
 Bayesian analysis produces the  working predictive distribution
\(D_k(Y_k)\) based on  the collection of the MCMC samples of \(Y_k\). Specifically,  $(\beta_0,\beta_1) ~ind~ \mbox{N}(0, 100)$  and  $b \sim \mbox{Uniform}(0,5)$. A total of $2000$ iterations were obtained with $1000$ burn-in. The convergence was checked visually by the trace plot as well as the Gelman-Rubin convergence statistic \citep{gelman1992inference}.

\subsection{Results}
\label{sec:org944f88a}
Based on a single replication of \(K=1,\!000\) sampling units, Figure~\ref{fig:orgf968955}
compares \(R^*\) and \(R^\ddagger\) from the working model \(D\) under the true
models \(F^{(0)}\) (the null hypothesis) and \(F^{(1)}\) (the
alternative hypothesis).  Note that under the null (the left column), the distribution  of \(R^\ddagger\) is very close to the \(N(0,1)\) reference, which is not the case for  \(R^*\), indicating  that $R^*$ is poorly calibrated.  
Also, under the alternative (right column)  $R^\ddag$ is more likely to detect a model discrepancy than is $R^*$. 

\begin{figure}[htbp]
\centering
\includegraphics[width=1\textwidth]{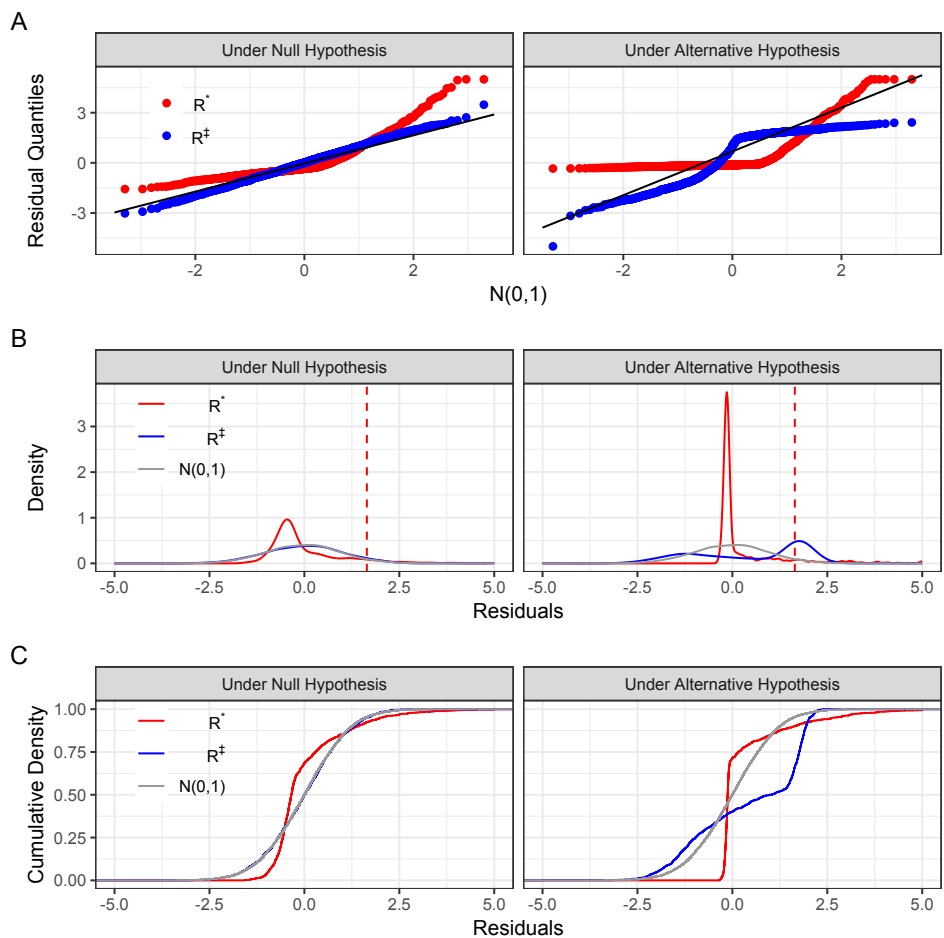}
\caption{\label{fig:orgf968955} 
\(R^\ddagger\) and \(R^*\) based on a single replication of \(1,000\) sampling units.  Panel A,  Q-Q plots; panel B, smoothed densities; Panel C, empirical CDFs. Panel A includes the reference $y=x$ line (black). Panels B
and C include   the reference \(N(0,1)\) density and  distribution.  The vertical dashed line in Panel B corresponds to the right-sided rejection region for \(\alpha = 0.05\). The reference \(N(0,1)\) line overlaps with \(R^\ddagger\) in the left column (null hypothesis) in Panels B and C.}
\end{figure}

\newpage

Table~\ref{tbl:simu4} reports  the estimated rejection rate based on $1,\!000$ replications  of the null hypothesis where the {\em true}  residual of a sampling unit is $0$, with nominal, right-sided 
$\alpha= .05$  for
\(R^\ddagger\), \(R^*\), and calibrated \(R^*\) with  \(\alpha^*\) from equation~\ref{eq:calibrr}.  Type~I
error is inflated for \(R^*\), but neither for \(R^\ddagger\) nor for calibrated \(R^*.\)
All simulations produce similar comparisons.

\begin{table}[ht]
  \centering
  \begin{tabular}{crrccccr}
    \hline\hline
    True  &           &   & \multicolumn{3}{c}{Rejection Rate} & Calibrated \\ \cline{4-6}
    Model &Hypothesis & N & $R^*$ & $R^*$ (calibrated) & $R^\ddagger$ & $\alpha^*(0.05)$ \\ 
    \hline
   $F^{(0)}$ & Null         &  150 & 0.074 & 0.050 & 0.049 & 0.026 \\ 
            &              &  175 &  0.074 & 0.050 & 0.050 & 0.026\\ 
            &              &  200 &  0.074 & 0.050 & 0.050 & 0.026\\ 
            &              &  225 &  0.074 & 0.050 & 0.050 & 0.026\\ 
            &              &  250 &  0.073 & 0.050 & 0.049 & 0.026\\  \hline 
   $F^{(1)}$ &  Alternative &  150 & 0.103 & 0.324 & 0.321 & 0.428 \\
            &              &  175 &  0.103 & 0.326 & 0.323 & 0.429\\ 
            &              &  200 &  0.102 & 0.326 & 0.324 & 0.430\\ 
            &              &  225 &  0.102 & 0.319 & 0.316 & 0.424\\ 
            &              &  250 &  0.102 & 0.323 & 0.320 & 0.426\\ 
    \hline\hline
  \end{tabular}
  \caption{Estimated rejection rate based on $1,000$ replications with right-sided $\alpha = 0.05$ for $R^\ddagger$, $R^*$, and calibrated
    $R^*$ using  $\alpha^*$ from equation~\ref{eq:calibrr}.}
  \label{tbl:simu4}
\end{table}


\section{Application to Protein Microarrays} \label{proteinMicroarray.sec} 

When a protein microarray is assembled, probes or specific proteins are arranged in rows and columns on a glass slide. After the sample of interest has been loaded onto the slide,  the array is scanned and light of different wavelengths produces a signal at each probe the intensity of which depends on the presence and quantity of a particular target protein in the sample. The scanning apparatus for protein microarrays produces two measurements at each probe: an observed foreground, signal $Y_{fg}$ and an observed background signal, $Y_{bg}$. We present simulated arrays composed of $10,\!000$ individual foreground and background signals generated under three conditions. Generally, the goal of this simulation is to compare a case where the analysis model incorporates more variability than is in the data generating model and a case where the analysis model and the data generating model match exactly. In this sense, we aim to show how percentile-based residuals can effectively assess general model fit. 

\subsection{Model and estimation}
\label{model_and_est_pm.sec}
In this simulation, we posit that  that a true underlying foreground signal, $S$, and a true underlying  background signal, $B$, multiply to produce a true, underlying quantity, $S\times B$. In the protein array, $B$ and $S \times B$ are  measured with multiplicative   errors $e_{bg}$ and $e_{fg}$ to produce $Y_{bg}$ and $Y_{fg}$.  Measurement errors  $e_{fg}$ and $e_{bg}$ are   independent, log-normal with  $\sis_{bg}, \sis_{fg}$ the respective variances of the underlying normal distributions.

These assumptions are encoded in the Bayesian hierarchical model
(the  analysis model, $D$),  
\bea \label{bayes.eqn}
Y_{bg}|B&\sim&\mbox{log-normal} \left \{\log(B)- \frac{\sis_{bg}}{2}, \sis_{bg}\right \}\nn\\
Y_{fg}|S,B &\sim& \mbox{log-normal} \left \{\log(B\times S)- \frac{\sis_{fg}}{2}, \sis_{fg}\right \}\\
S &\sim& \mbox{gamma} (\alpha_{s}, \beta_{s}), \alpha_s=\frac{\mu_s^2}{\sis_s}, \beta_s=\frac{\mu_s}{\sis_s} \nn \\
B &\sim& \mbox{gamma}  (\alpha_{b}, \beta_{b}), \alpha_b=\frac{\mu_b^2}{\sis_b}, \beta_b=\frac{\mu_b}{\sis_b} \nn \\
\mu_s, \mu_b &\sim& \mbox{uniform} (0, 10^6) \nn\\
\sis_s , \sis_b &\sim& \mbox{uniform} (0,10^8) \nn
\eea

We  performed simulations by generating $10,\!000$, or an array's worth of foreground and background signals on a three-tiered basis for the true distribution ($F$). Tier~1  involves fixing a single $S=6,\!000$ and $B=60$, Tier~2  involves fixing $\alpha_{s}= \alpha_{b}= \beta_{s}=\beta_{b}=1$ (not shown below), and Tier~3 involves randomly generating $\alpha_{s}, \alpha_{b}, \beta_{s}, \beta_{b}$.   The analysis model ($D$, equation~\ref{bayes.eqn}) matches Tier~3. In all tiers,  $\sis_{bg}=\sis_{fg}=0.1$.

We report results for Tier~1 and  Tier~3   (results for Tier~2 are very similar to those for Tier~1).  The analysis model is used to generate MCMC draws from relevant posterior distributions.  For  each MCMC, three chains were run with randomly selected initial values for all parameters with a  burn-in of $2,\!000$ and a subsequent  $10,\!000$ iterations with  a thinning interval of length 
$10$, resulting in a sample size of $1,\!000$ for each of the $10,\!000$ simulated probes on an array. The convergence was checked visually by the trace plot as well as the Gelman-Rubin convergence statistic \citep{gelman1992inference}.

\subsection{Analysis of residuals}\label{residual_analysis_pm.sec}

We define the standard residual and the percentile-based residual in this context as follows:\\
\bea \label{y_fg_residuals.eqn}
R_i^*&=& \frac{Y_{fg,i} - E[Y_{fg,i}|\mbox{model}(D)]}{\sqrt{Var[Y_{fg,i}|\mbox{model}(D)]}}\nn\\
R_i^{\ddag}&=& \Phi_{0,1}^{-1}[\mbox{percentile location of} Y_{fg,i} \mbox{in Pred}_i\{\mbox{model}(D)\}]\nn\\
\eea

In this case, `$D$' is all generated data points, $\mbox{model}(D)$ is the model described in~\ref{bayes.eqn} estimated using $D$, and $\mbox{Pred}_i\{\mbox{model}(D)\}$ is the full predictive distribution of $[Y_{fg,i}|\mbox{model}(D)]$. Note that best practice would be to set aside the $ith$ observation, fit the model and use it to predict the $ith$ data value using model$(D_{-i})$. However, with the large sample size in our example, the difference is negligible. 

While in Tier~3, the analytic and data-generating models match, in 
Tier 1, the analytic model brings in  more stochastic elements than are present in the data.  This discrepancy is clearly displayed in panel~B of Figure~\ref{unif_plots}.   The percentile-based residuals of the true $Y_{fg}$ values in the posterior predictive distribution $Y_{fg,i}|\mbox{model}(D)$
have  a standard normal distribution in Tier~3 , whereas the percentile-based residuals in Tier~1 have a much narrower spread and are centered slightly to the left of zero. 

\begin{figure}[ht]
\centering
{\includegraphics[width=150mm, height=150mm]{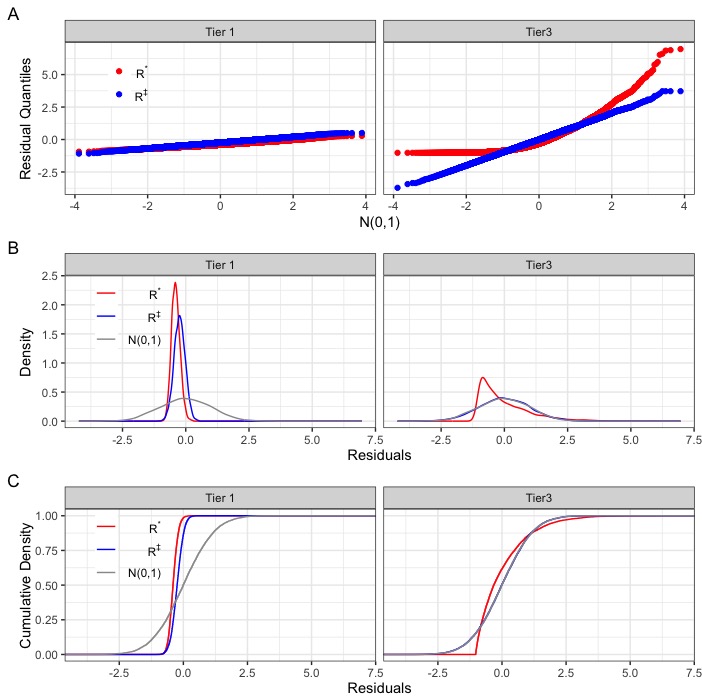}}
\caption{\label{unif_plots} Distribution of the percentile location of $Y_fg$  in its full posterior predictive distribution, Tier~1 (left column), Tier~3 (right column).  Panel A,  Q-Q plots; panel B, smoothed densities; Panel C, empirical CDFs. Panel A includes the reference $y=x$ line (black). Panels B
and C include   the reference \(N(0,1)\) density and  distribution.}
\end{figure}

Importantly, in Tier~3, where the data analysis model and the data generation model match exactly, the standard residuals, $R^*$ depart from the standard normal distribution. As evidenced clearly in Panel~A of figure~\ref{unif_plots} the quantiles of the standard residuals as compared to the $N(0,1)$ quantiles would lead to an inflated Type I error rate, that is a higher probability of erroneously detecting a discrepancy between the true predictive distribution and the working predictive distribution.


\section{Sub-national estimates of contraceptive use} \label{pmaback.sec}
Access to family planning provides multi-faceted benefits to women and their families. It has been shown to reduce maternal and child mortality, empower women and girls, and enhance environmental sustainability. Remarkable progress has been made in the past several decades in measuring family planning use rates around the world, and in low-income countries in particular. However, the national focus of those surveys (e.g., Demographic and Health Surveys) makes the result impractical for monitoring and evaluation at sub-national levels. Given the fact that many policies are made and implemented at subnational levels, there is an urgent need to reach out to policy makers at local levels and empower them with relevant estimates. 

\subsection{The Performance Monitoring and Evaluation 2020 Survey} \label{pmadata.sec} 

The data analyzed here are drawn from Performance Monitoring and Evaluation 2020 (PMA2020) survey conducted in Kenya during May-July of 2014. The survey interviewed 3479 women in 111 enumeration areas (EAs) in Kenya. Funded by the Bill and Melinda Gates Foundation, PMA2020 was originally designed to facilitate annual progress reporting in support of the goals and principles of Family Planing 2020 (FP2020) initiative across priority countries in Africa and Asia \citep{Zimmerman:etal:2017}. The survey uses mobile devices (smartphones) to routinely gather nationally representative data on key family planning indicators. Data are collected at the woman, household, and facility levels by a network of resident enumerators stationed throughout the country.  
\bigbreak

Given its original goal of providing national estimates, PMA2020 surveys are adequately powered to provide reliable estimates for the whole nation,and in some cases for urban and rural regions separately. While national estimates become more available, regional, county and district officials expressed interests in estimates of indicators to monitor and evaluate progress at their level of operation and responsibilities. To meet this need, a Bayesian hierarchical model was developed for several African countries with multiple rounds of PMA surveys
\citep{Li:etal:2018}.  Described below is a model for a single cross-section  of the repeated cross-sectional data collected  in that study. The sufficiently large number of women per EA (about 30) makes the standard and percentile-based residuals nearly identical, and to illustrate the difference between the two, we present results  based on a 30\%, unstratified random sample of the original PMA dataset.

\bigbreak

\subsection{Bayesian Modeling} \label{pmamodel.sec} 
$Y_{ik}=0/1$ is the indicator of woman $i$ in EA $k$
$(i = 1, \ldots, n_k; k = 1, \ldots K)$ not using (0) or using (1) contraceptive methods. 
 $X_{ik}$  is the woman-sepcific vector of covariates (e.g., age, education, parity), and $U_k\sim N(0, \tausq)$ is an  EA-specific random effect. The working model is, 
\begin{eqnarray*}
P_{ik} &=& \pr(Y_{ik}=1 \mid X_{ik},\beta, U_k = u_k) \\
    \logit\left (P_{ik} \right )  &=&X_{ik}\beta + u_{k}.
\end{eqnarray*}
The $P_{ik}$ are `rolled up' to the EA level (producing
$P_{+k}$), and these are the model-based, EA-specific
 rates.

As described in \cite{Li:etal:2018},  we used  Markov chain Monte Carlo implemented using JAGS (version 4.3.0) and  conducted in R, to generate  draws from the joint posterior distribution of $(\beta, \tausq; U_1, \ldots, U_K)$, and consequently for $P_{ik}$ and $P_{+k}$, the latter being `Small Area Estimates.' 
The  quantity 
`Avelogistic' is produced by mixing the $P_{ik}$  over the posterior for $(\beta, \tausq)$, but over the mixture of  $N(0, \tausq)$ priors
for $U_k$.  This approach produces residuals relative to a population model rather than those relative to a model tuned to the EA, and so evaluates model adequacy.

\subsection{Data Analysis} \label{pmaresults.sec} 
Figure~\ref{fig:pmafig2} shows the percentile-based and standardized residuals for the 111 EAs. The variation in percentile-based residuals is smaller than in standard residuals (Panel A). Panel B suggests that both  $R^*$ and $R^\ddag$ are close to normally distributed, that each set has a variance greater than 1.0, with the percentile-based having a smaller spread. Panel C reveals additional detail, showing  that  the $R^*$ are close to Gaussian (albeit with a large spread), but the $R^\ddag$ are bimodal, possibly indicating that a binary covariate is missing from the model. We haven't been able to find a covariate that removes the bimodality, but note that
$R^\ddag$ identifies the problem, the $R^*$ do not.


\begin{figure}[ht]
\centering
\includegraphics[scale=0.8]{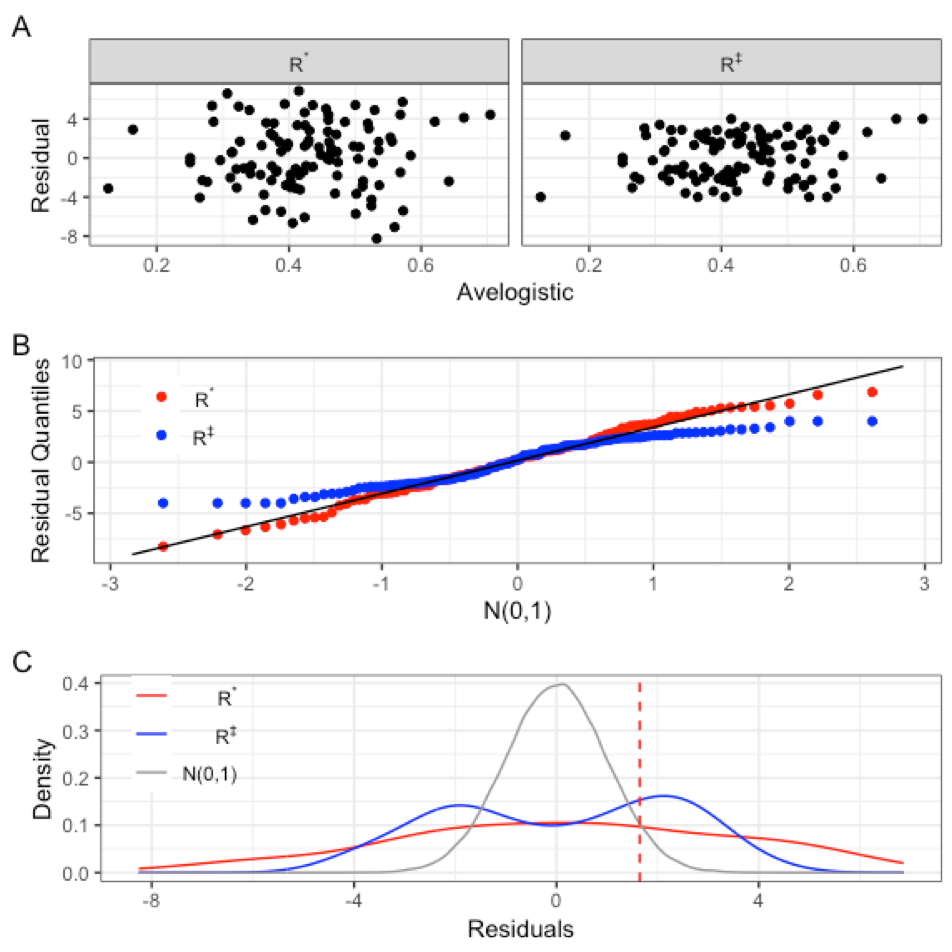}
\caption{(A) $R^*$ (left panel) and $R^\ddag$ (right panel)  vs avelogistic; (B) Q-Q Plots for \(R^\ddag\) and \(R^*\)  residuals; (C) Histograms of \(R^\ddag\) and \(R^*\) residuals, along with the reference $N(0,1)$ density.}
\label{fig:pmafig2}
\end{figure}


\clearpage
\newpage

\section{Discussion}
Residuals are a mainstay for assessing model fit and detecting  outliers.  We define and evaluate a percentile-based approach that, by respecting all aspects of a predictive distribution, is a considerable improvement over  use of the standard residuals.
Improvements include properly calibrated Type~I error when testing a working hypothesis, improved power, and more revealing diagnostic plots.   While it is the case that the nominal $\alpha$ for testing can be adjusted to calibrate the standard residuals to have a desired Type~I error, and with careful adjustments can improve residual plots and other model diagnostics, the percentile-approach automatically takes care of these issues.
The percentile-based approach has the added benefit of encouraging development of a full predictive distribution that incorporates sampling, measurement, and  modeling-induced uncertainties.  Consequently,  we encourage its use.

{\bf Funding}\\
The authors gratefully acknowledge partial support  from the following sources: SB, NIH-NIAID, U19-AI089680; QL, Performance Monitoring and Evaluation 2020 project from the Bill \& Melinda Gates Foundation; CW, NCI P30CA006973, NCI 1 P50 CA098252; TAL, NIH-NIAID, U19-AI089680

{\bf Acknowledgments}\\
The authors are grateful to the Southern and Central Africa International Centers for Excellence in Malaria Research for useful discussions particularly in relation to the section on applications to protein microarrays.

{\raggedright
\section*{Bibliography} \label{biblio} 
\renewcommand\refname{} 
\bibliographystyle{jasa}
\bibliography{percentiles}    
}

\end{document}